\documentclass[11pt,conference,onecolumn]{IEEEtran}
\usepackage{amssymb}
\usepackage{amsthm}
\usepackage{amsmath}
\usepackage{amsfonts}
\usepackage{graphicx}
\usepackage{graphics}
\usepackage{subfigure}
\usepackage{multirow}
\usepackage{anysize}
\marginsize{0.4in}{0.4in}{0.5in}{1in}

\newcommand{\limi}[1]{\;\dot{#1}\;}

\newtheorem{thm}{\it{Theorem}}

\newtheorem{Def}{\it{Definition}}

\title{On the Average Complexity of Sphere Decoding in Lattice Space-Time Coded MIMO Channel}

\author{
\authorblockN{Walid Abediseid}
\authorblockA{Physical Sciences and Engineering Division\\
King Abdullah University for Science and Technology (KAUST) \\
Thuwal, Saudi Arabia \\
walid.abediseid@kaust.edu.sa}
}

\begin{document}
\maketitle
\begin{abstract}
The exact average complexity analysis of the basic sphere decoder for general space-time codes applied to multiple-input multiple-output (MIMO) wireless channel is known to be difficult. In this work, we shed the light on the computational complexity of sphere decoding for the quasi-static, LAttice Space-Time (LAST) coded MIMO channel. Specifically, we drive an upper bound of the tail distribution of the decoder's computational complexity. We show that, when the computational complexity exceeds a certain limit, this upper bound becomes dominated by the outage probability achieved by LAST coding and sphere decoding schemes. We then calculate the minimum average computational complexity that is required by the decoder to achieve near optimal performance in terms of the system parameters. Our results indicate that there exists a \textit{cut-off} rate (multiplexing gain) for which the average complexity remains bounded.
\end{abstract}
\section{Introduction}
Since its introduction to multiple-input multiple-output (MIMO) wireless communication systems, the sphere decoder has become an attractive efficient implementation of the maximum-likelihood (ML) decoder, especially for small signal dimensions and/or moderate to large SNRs. Such a decoder allows for significant reduction in (average) decoding complexity as opposed to the ML decoder without sacrificing performance. In general, the sphere decoder \cite{HV}--\cite{SJ} is commonly used in communication systems that can be well-described by the following (real) \textit{linear Gaussian vector channel} model
\begin{equation}\label{linear_model}
\pmb{y}=\pmb{Mx}+\pmb{e},
\end{equation}
where $\pmb{x}\in\mathbb{R}^m$ is the input to the channel, $\pmb{y}\in\mathbb{R}^n$ is the output of the channel, $\pmb{e}\in\mathbb{R}^n$ is the additive Gaussian noise vector with entries that are independent identically distributed, zero-mean Gaussian random variables with variance $1/2$, and $\pmb{M}\in\mathbb{R}^{n\times m}$ is a matrix representing the channel linear mapping.

The input-output relation describing the channel that is given in (\ref{linear_model}) allows for the use of \textit{lattice theory} \cite{Conway} to analyze many digital communication systems. In this paper, we assume that $\pmb{x}$ is a codeword selected from a lattice code. Let $\Lambda_c\stackrel{\Delta}{=}\Lambda(\pmb{G})=\{\pmb{x}=\pmb{G}\pmb{z}:\pmb{z}\in\mathbb{Z}^m\}$ be a lattice in $\mathbb{R}^m$ where $\pmb{G}$ is an $m\times m$ full-rank lattice generator matrix. The Voronoi cell, $\mathcal{V}_{\pmb{x}}(\pmb{G})$, that corresponds to the lattice point $\pmb{x}\in\Lambda_c$ is the set of points in $\mathbb{R}^m$ closest to $\pmb{x}$ than to any other point $\pmb{\lambda}\in\Lambda_c$, with volume that is given by $V_c\triangleq {\rm Vol}(\mathcal{V}_{\pmb{x}}(\pmb{G}))=\sqrt{{\rm det}(\pmb{G}^{\mathsf{T}}\pmb{G})}$. An $m$-dimensional lattice code $\mathcal{C}(\Lambda_c,\pmb{u}_o,\mathcal{R})$ is the finite subset of the lattice translate $\Lambda_c+\pmb{u}_0$ inside the shaping region $\mathcal{R}$, i.e., $\mathcal{C}=\{\Lambda_c+\pmb{u}_0\}\cap\mathcal{R}$, where $\mathcal{R}$ is a bounded measurable region\footnote{In this paper, we consider a shaping region $\mathcal{R}$ that corresponds to the Voronoi cell $\mathcal{V}_s$ of a sublattice $\Lambda_s$ of $\Lambda_c$, i.e., $\Lambda_s\subseteq \Lambda_c$. The generated codes are called nested (or Voronoi) lattice codes (see \cite{EZ2} for more details).} of $\mathbb{R}^m$.

Space-time codes based on lattices have been used in MIMO channels due to their low encoding complexity (e.g., nested or Voronoi codes \cite{Conway}) and the capability of achieving excellent error performance (see \cite{GC}). Another important aspect of using lattice space-time (LAST) codes is that they can be decoded by a class of efficient decoders known as  \textit{lattice decoders}. These decoder algorithms reduce complexity by relaxing the code boundary constraint and find the point of the underlying (infinite) lattice closest to the received point, i.e.,
\begin{equation}\label{lattice_decoding}
\hat{\pmb{x}}={\rm arg}\min_{\pmb{x}\in{\Lambda_c}}\|\pmb{y}-\pmb{M}\pmb{x}\|^2.
\end{equation}
It is well-known that sphere decoding based on Fincke-Pohst and Schnorr-Euchner enumerations are efficient strategies to solve (\ref{lattice_decoding}) and have been widely used for signal detection in MIMO systems of small dimensions (see \cite{Label4} and references therein). In this work, and for the sake of simplicity, we consider the Fincke-Pohst sphere decoder without radius reduction and restarting (see \cite{Label4} for more details).

\subsection{Previous Work}
Up to date, most of the work on the complexity of sphere decoding focused on characterizing the mean and the variance of the decoder's complexity, particularly for the \textit{uncoded} MIMO channel (e.g., V-BLAST) \cite{HV}--\cite{JB}. Seethaler \textit{et. al.} \cite{SJ} considered the derivation of the computational distribution of the sphere decoder for the $M\times N$ uncoded MIMO channel. Characterizing and understanding the complexity distribution is important, especially when the sphere decoder is used under practically relevant runtime constraints. The computational tail distribution is defined as $\Pr(C\geq L)$, where $C$ is the overall decoding complexity, and $L$ is the distribution parameter. It has been shown in \cite{SJ} that. as $L\rightarrow\infty$, the computational tail distribution follows a Pareto-type with tail exponent given by $N-M+1$, i.e.,
$$\Pr(C\geq L)=L^{-(N-M+1)},\quad L\rightarrow\infty.$$

However, the main drawback of their work is that they consider the decoder's complexity analysis when the number of computations performed by the decoder increases without bound. In other words, although the behaviour of the tail distribution is characterized, they do not specify a value of $L$ to indicate when the computations become excessive. In fact, when the decoder is used under practically runtime constraint, it is sometimes desirable to allow the decoder to terminate the search once a limit is exceeded and declare an error. 
%Moreover, the \textit{exact} average complexity of the sphere decoder when applied to the uncoded MIMO channel was not studied.

Achieving higher diversity and multiplexing gains require incorporating error control coding (across antenna and time) at the transmitter. Several works have considered the computational complexity analysis of optimal and sub-optimal decoders applied to LAST coded $M\times N$ MIMO channels \cite{Label4}--\cite{JElia}. A first attempt toward specifying the exact complexity required by the decoder to achieve the optimal diversity-multiplexing tradeoff (DMT) \cite{ZT} of the quasi-static LAST coded MIMO channel was considered in \cite{EJ}. It was shown that the optimal tradeoff can be achieved using lattice reduction (LR) aided linear decoders at a worst-case complexity $O(\log\rho)$, where $\rho$ is the SNR. This corresponds to a linear increase in complexity as a function of the code rate $R$ at the high SNR regime, where $R=r\log\rho$ with $r\leq\min\{M,N\}$ referred to as the multiplexing gain of the coding scheme. However, this very low decoding complexity comes at the expense of a large gap from the sphere decoder's error performance. In order to close this gap, lattice sequential decoding algorithms \cite{MGDC} are considered efficient decoders that achieve good error performance with much lower decoding complexity. In \cite{WD}, we have analysed in details the computational tail distribution and the average complexity of the lattice sequential decoder. Specifically, we have shown that when the computational complexity exceeds a certain limit, the tail distribution becomes upper bounded by the outage probability achieved by LAST coding and sequential decoding schemes, i.e.,
$$\Pr(C\geq L)\leq \rho^{-d_{\rm out}(r)},\quad L\geq L_0,$$
where $d_{\rm out}(r)$ is the DMT achieved by the underlying coding and the decoding schemes. This interesting result indicates that one may save on decoding complexity while still achieving near-outage performance by setting a \textit{time-out} limit at the decoder so that when the computational complexity exceeds the limit, the decoder terminates the search and declare an error. In this work, and for the sake of completeness, we study the complexity behaviour of the optimal sphere decoder in the quasi-static MIMO channel. This would allow us to compare the complexity of the sphere decoder with other low complexity decoders and see whether it is worth sacrificing performance for complexity achieved by such decoders.

The work in \cite{JElia} considers sphere decoding algorithms that describe exact ML decoding. Such algorithms take into consideration the coding boundary constraint to achieve the ML performance, while in this paper we consider sphere decoding algorithms that perform lattice decoding. To be specific, it has been shown in \cite{JElia} that the total number of computations required by the sphere decoder to achieve a vanishing gap from the ML performance is given asymptotically (at high SNR) by $\rho^{c(r)}$, where $c(r)={Tr}(1-{r/ M})$, $\forall 0\leq r\leq M$, has been defined as the sphere decoder's SNR complexity exponent. Compared to the exhaustive ML decoder's complexity $\rho^{Tr}$, there is a $(1-r/M)$ reduction factor in the complexity SNR exponent achieved by the sphere decoder without sacrificing the performance. However, two important aspects about the complexity behaviour of the sphere decoder have not been considered in \cite{JElia}. First, at multiplexing gain $r=0$, i.e., at fixed coding rates $R$, one cannot realize the complexity saving advantage achieved by the sphere decoder using the above definition of the computational complexity. This is very important, since most of the experiments on the complexity of various sphere decoding algorithms are performed using channel codes with fixed rates. Moreover, the comparisons between several decoders applied to the outage-limited MIMO channel are usually performed experimentally through their corresponding average decoding complexity, a topic that was not considered in \cite{JElia}.

 \subsection{Main Contribution}
 The main contribution in this work focuses on the complexity tail distribution and the (exact) average computational complexity of the sphere decoder for the LAST coded MIMO channel. Specifically, we consider the analysis of the decoder when minimum mean-square error decision-feedback equalization (MMSE-DFE) is applied. We derive the asymptotic average complexity of the decoder in terms of the system parameters: the SNR $\rho$, the number of transmit antennas $M$, the number of receive antennas $N$, and the codeword length $T$. For this type of decoding, we specify the required system parameters that are needed to achieve the corresponding DMT with fairly low decoding complexity. In general, it is shown that the sphere decoder has much lower \textit{asymptotic} complexity than the exhaustive ML decoder. Moreover, we show that there exists a \textit{cut-off} rate (multiplexing gain) for which the average complexity remains bounded. Simulation results are used to verify our theoretical analysis and claims.

Throughout the paper, we use the following notation. The superscript $^c$ denotes complex quantities, $^\mathsf{T}$ denotes transpose, and $^\mathsf{H}$ denotes Hermitian transpose. We refer to $g(z)\limi{=}z^a$ as $\lim_{z\rightarrow\infty}g(z)/\log(z)=a$, $\dot{\geq}$ and $\dot{\leq}$ are used similarly. For a bounded Jordan-measurable region $\mathcal{R}\subset\mathbb{R}^m$, $V(\mathcal{R})$ denotes the volume of $\mathcal{R}$.  We denote $\mathcal{S}_{\pmb{x}}^m(r)$ by the $m$-dimensional hypersphere of radius $r$ centred at $\pmb{x}$ with $V(\mathcal{S}_{\pmb{x}}^m(r))=(\pi r^2)^{m/2}/\Gamma(m/2+1)$, and $\pmb{I}_m$ denotes the $m\times m$ identity matrix.  The notation $\pmb{v}\sim\mathcal{N}(\pmb{\mu},\pmb{K})$ indicates that $\pmb{v}$ is a real Gaussian random vector with mean $\pmb{\mu}$ and covariance matrix~$\pmb{K}$. The complement of a set $\mathcal{A}$ is denoted by $\overline{\mathcal{A}}$.

\section{LAST Coding and Lattice Decoding}
We consider a quasi-static, Rayleigh fading MIMO channel with $M$-transmit, $N$-receive antennas, and no channel state information (CSI) at the transmitter and perfect CSI at the receiver. The complex base-band model of the received signal can be mathematically described by (for $T$ channel uses)
\begin{equation}
\pmb{Y}^c=\sqrt{\rho}\pmb{H}^c\pmb{X}^c+\pmb{W}^c,
\label{channel}
\end{equation}
where $\pmb{X}^c\in\mathbb{C}^{M\times T}$ is the transmitted space-time code matrix, $\pmb{Y}^c\in\mathbb{C}^{N\times T}$ is the received signal matrix, $\pmb{W}^c\in\mathbb{C}^{N\times T}$ is the noise matrix, $\pmb{H}^c\in\mathbb{C}^{N\times M}$ is the channel matrix, and $\rho=\mathsf{SNR}/M$ is the normalized SNR at each receive antenna with respect to $M$. The elements of both the noise matrix and the channel fading gain matrix are assumed to be independent identically distributed (i.i.d.) zero mean circularly symmetric complex Gaussian random variables with variance $\sigma^2=1$.

An $M\times T$ space-time coding scheme is a full-dimensional LAttice Space-Time (LAST) code if its vectorized (real) codebook (corresponding to the channel model (\ref{linear_model})) is a lattice code with dimension $m=2MT$. As discussed in \cite{GC}, the design of space-time signals reduces to the construction of a codebook $\mathcal{C}\subseteq\mathbb{R}^{2MT}$ with code rate $R={1\over T}\log|\mathcal{C}|$, satisfying the input averaging power constraint
$${1\over|\mathcal{C}|}\sum_{\pmb{x}\in\mathcal{C}}\|\pmb{x}\|^2\leq MT.$$

When lattice decoding is pre-processed by MMSE-DFE filtering\footnote{Here, we perform the QR-decomposition on the \textit{augmented} channel matrix
$$\tilde{\pmb{H}}=\begin{pmatrix} \pmb{H}\cr \pmb{I}\end{pmatrix}=\tilde{\pmb{Q}}\pmb{R},$$
where $\pmb{H}$ is real-valued equivalent channel gain matrix, $\tilde{\pmb{Q}}\in\mathbb{R}^{(n+m)\times m}$ has orthonormal columns, and $\pmb{R}\in\mathbb{R}^{m\times m}$ is an upper triangular with positive diagonal elements. If we let $\pmb{Q}=\pmb{H}\pmb{R}^{-1}$ the upper $n\times m$ part of $\pmb{Q}$, then the matrices $\pmb{F}=\pmb{Q}^\mathsf{T}$ and $\pmb{B}=\pmb{R}$ are called the MMSE-DFE \textit{forward} and \textit{backward} filters, respectively. At the receiver, the received signal, $\pmb{y}$, is multiplied by the forward filter matrix $\pmb{F}$ of the MMSE-DFE to get $\pmb{y}=\pmb{B}\pmb{x}+\pmb{e}$. This is equivalent to (\ref{linear_model}) with $\pmb{M}=\pmb{B}$ where $\pmb{B}$ has the property that $\det(\pmb{B}^\mathsf{T}\pmb{B})=\left[\det\left(\pmb{I}_M+{\rho}(\pmb{H}^c)^\mathsf{H}\pmb{H}^c\right)\right]^{2T}$ (refer to \cite{GC} for more details about this topic).}, the equivalent real model of the above channel can be easily shown to be given by (\ref{linear_model}) with $\pmb{M}$ that satisfies (see \cite{GC} for more details)
\begin{equation}\label{MMSE_Decoding}
\det(\pmb{M}^\mathsf{T}\pmb{M})=\left[\det\left(\pmb{I}_M+{\rho}(\pmb{H}^c)^\mathsf{H}\pmb{H}^c\right)\right]^{2T},
\end{equation}

\begin{Def}
 Consider a family of LAST codes $\mathcal{C}_{\rho}$ for fixed $M$ and $T$, obtained from lattices of a given dimension $m=2MT$ and indexed by their operating SNR $\rho$. The code $\mathcal{C}_{\rho}$ has rate $R(\rho)$, average error probability $P_e(\rho)$, and decoding computational complexity $\Pr(C\geq L)$ (averaged over the random channel matrix $\pmb{H}^c$). The multiplexing gain, diversity order, and complexity tail exponent are defined respectively as
 \begin{equation*}
 r=\lim_{\rho\rightarrow\infty}{R(\rho)\over\log{\rho}},\quad d=\lim_{\rho\rightarrow\infty}{-\log P_e(\rho)\over\log{\rho}},\quad\eta=\lim_{\rho\rightarrow\infty}{-\log \Pr(C\geq L)\over\log \rho}.
\end{equation*}
\end{Def}
It has been shown in \cite{GC} that LAST coding and lattice decoding can achieve rates up to
\begin{equation}\label{achievable}
R_{\rm LAST}(\rho,\pmb{H}^c)=\log\det(\pmb{M}^\mathsf{T}\pmb{M})^{1/2T}=\log\det\left(\pmb{I}_M+{\rho}(\pmb{H}^c)^\mathsf{H}\pmb{H}^c\right).
\end{equation}
For the underlying quasi-static MIMO channel, it is well-known that the asymptotic error performance, $P_e(\rho)$, of any coding and decoding schemes is dominated by the \textit{outage probability}, $P_{\rm out}(\rho,R)$, i.e., $P_e(\rho)\limi{=}P_{\rm out}(\rho,R)$. For LAST coding and lattice decoding schemes, the outage probability is defined by
\begin{equation}
P_{\rm out}(\rho,R)=\Pr(R\geq R_{\rm LAST}(\rho,\pmb{H}_c))\limi{=}\rho^{-d^*_{\rm out}(r)},
\end{equation}
where $d^*_{\rm out}(r)= (M-r)(N-r)$, $\forall\;r\in[0,\min\{M,N\})$, is defined as the \textit{optimal} diversity-multiplexing tradeoff (DMT) achieved by the channel \cite{ZT}.

Define the outage event $\mathcal{O}(\rho)=\{\pmb{H}^c:R(\rho)\geq R_{\rm LAST}(\rho,\pmb{H}^c)\}$, and denote the transmission rate $R(\rho)=r\log\rho$. Let $0\leq\lambda_1\leq\cdots\leq\lambda_{M}$ be the ordered eigenvalues of $(\pmb{H}^c)^\mathsf{H}{\pmb{H}^c}$, and define $\pmb{\alpha}=(\alpha_1,\cdots,\alpha_M)$, where $\alpha_i\triangleq-\log\lambda_i/\log\rho$. As discussed in \cite{ZT}, at high SNR, the non-negative values of $\pmb{\alpha}$ only contributes to the outage event. Therefore, the outage event can be expressed as
\begin{equation}\label{outage_MMSE}
\mathcal{O}\limi{=}\left\{\pmb{\alpha}\in\mathbb{R}_+^{M}:\; \sum_{i=1}^{M}(1-\alpha_i)^+<r\right\}.
\end{equation}

In what follows, we summarize the results derived in \cite{GC}. For the  lattice decoding, there exists a sequence of full-dimensional LAST codes that achieves DMT 
\begin{equation}\label{MMSE-DFE}
d_{\rm out}^*(r)=(M-r)(N-r),\quad \forall r\in [0,\min\{M,N\}],
\end{equation}
under the constraint $T\geq M+N-1$ (see \cite{GC} for more details).

\section{Lattice Decoding via Sphere Decoding}
While ML decoding performs exhaustive search over all codewords $\pmb{c}\in\mathcal{C}(\Lambda_c,\mathcal{R})$, sphere decoding algorithms find the closest lattice point $\pmb{x}\in\Lambda_c$ to the received signal $\pmb{y}$ within a sphere radius $R_s$ centred at the received signal. It is well-known that the sphere decoder allows for significant reduction in decoding complexity for small signal
dimensions and average to large values of SNR. Depending on whether the sphere decoder incorporate the boundaries of the lattice code (i.e., $\mathcal{R}$) into the search algorithm or not, one can achieve ML or \textit{near}-ML performance. Here, we consider sphere decoding algorithms that describe lattice decoding, i.e., the class of decoding algorithms that do not take into account the shaping region $\mathcal{R}$.

In general, the sphere decoder, after QR decomposition of the channel-code matrix $\pmb{MG}=\pmb{QR}$, finds all integer lattice points $\pmb{z}\in\mathbb{Z}^m$ that satisfy the sphere constraint
\begin{equation}\label{yy}
\|\pmb{y}'-\pmb{Rz}\|^2=\sum\limits_{k=1}^m\|{\pmb{y}'}_1^k-\pmb{R}_{kk}\pmb{z}_1^k\|^2\leq R_s^2,
\end{equation}
where $\pmb{y}'=\pmb{Q}^\mathsf{T}\pmb{y}$, $\pmb{Q}$ is an orthogonal matrix, and $\pmb{R}$ is an $m\times m$ upper triangular matrix with positive diagonal elements. Moreover, $\pmb{z}_1^k=[z_k,\cdots,z_2,z_1]^{\mathsf{T}}$ denotes the last $k$ components of the integer vector $\pmb{z}$, $\pmb{R}_{kk}$ is the lower $k\times k$ part of the matrix $\pmb{R}$, ${\pmb{y}'}_1^k$ is the last $k$ components of the vector ${\pmb{y}'}$. The structure of $\pmb{R}$ allows one to perform backward sequential search from layer (dimension) $m$ to layer $1$. Several algorithms were developed to efficiently perform the search (see for example \cite{Label4}). Once all points are listed, one can find the point that is closest in distance to $\pmb{y}$.

It is more convenient to look at the sphere decoder as a search in a \textit{tree} with $m$ layers. The $k$-th layer, where $1\leq k\leq m$, contains nodes that correspond to the partial integer lattice points $\pmb{z}_1^k\in\mathbb{Z}^k$. In this case, one may define the computational complexity of the sphere decoder as the total number of nodes that have been visited (or extended) by the decoder during the search.

Define the indicator function $\phi(\pmb{z}_1^k)$ by
\begin{equation}\label{phi}
\phi(\pmb{z}_1^k)=\begin{cases} 1, &\text{if $\pmb{z}_1^k$ is extended;}\cr
                                                          0, &\text{otherwise.}\end{cases}
                                                          \end{equation}
Then, the total number of partial integer lattice points $\pmb{z}_1^k\in\mathbb{Z}^k$ found by the decoder at layer $k$ can be expressed as
\begin{equation}\label{Ck}
C_k=\sum\limits_{\pmb{z}_1^k\in\mathbb{Z}^k}\phi(\pmb{z}_1^k).
\end{equation}
Therefore, the total computational complexity of the sphere decoder $C$ that is required to find the closest lattice point to the received signal is given by $C=\sum_{k=1}^m C_k$.

\subsection{Sphere Radius Selection}
The selection of the initial radius $R_s$ at the beginning of the search is of crucial importance in the computational complexity analysis. Choosing a small value of $R_s$ may result in finding no lattice points inside the sphere (i.e., $C_k=0$ for some $1\leq k\leq m$). On the other hand, choosing a very large value of $R_s$ results in finding too many lattice points inside the sphere that leads to very large computational complexity. As such, the sphere radius $R_s$ must be chosen sufficiently large for the search sphere to contain at least one lattice point.

Selecting $R_s=r_{\rm cov}(\pmb{MG})$, i.e., the covering radius\footnote{The covering radius $r_{\rm cov}(\pmb{G})$ of a lattice $\Lambda(\pmb{G})$ is the radius of the smallest sphere centred at the origin that contains $\mathcal{V}_{\pmb{0}}(\pmb{G})$.} of the lattice generated by $\pmb{MG}$, guarantees the existing of at least one lattice point inside the sphere. Unfortunately, the computation of $r_{\rm cov}$ for a general lattice is very difficult. Another choice of $R_s$ is the distance between the Babai estimate and the vector $\pmb{y}$. As mentioned in \cite{HV}, although this choice guarantees the existence of at least one lattice point (the Babai estimate) inside the sphere, it not clear in general whether it leads to too many lattice points inside the sphere.

In this work, we follow a different approach to find a sphere radius that guarantees the existing of at least one lattice point inside the sphere. This particular choice of the sphere radius is shown to simplify the analysis of deriving an upper bound on the decoder's computational complexity. The basic idea of this approach (as will be shown in the sequel) is to separate the typical noise events from the non-typical ones.  This allows the separation of the ``typical'' lattice points (lattice points that are highly likely to be generated by the sphere decoder) from the atypical ones.

\subsection{The $k$-th Layer Complexity}
In this section, we would like to provide some insight about the computational complexity of the sphere decoder at the $k$-th layer (more details can be found in \cite{SJ}). This may assist us in the derivation of the upper bound on the computational complexity distribution as will be shown in the sequel.

As mentioned previously, the computational complexity of the sphere decoder at the $k$-th layer is determined by the total number of partial lattice points $\pmb{z}_1^k\in\mathbb{Z}^k$ that satisfy the $k$-th layer sphere constraint
$$\|{\pmb{y}'}_1^k-\pmb{R}_{kk}\pmb{z}_1^k\|\leq R_s.$$
We assume that $R_s$ is chosen sufficiently large enough so that at least one lattice point is found inside the sphere (details on how $R_s$ is selected will be introduced next). It is clear that the computational complexity of the decoder depends on the distributions of $\pmb{y}_1^k$ and $\pmb{R}_{kk}$. Since those two quantities are random, the computational complexity analysis of the sphere decoder is considered difficult.

Most of the work on the complexity of sphere decoding (see \cite{Label10}, \cite{SJ}, \cite{BK}), rely on approximating $C_k$ by
\begin{equation}\label{aprrox_bound2}
C_k\approx {V(\mathcal{S}_{\pmb{0}}^k(R_s))\over \det(\pmb{R}_{kk}^\mathsf{T}\pmb{R}_{kk})^{1/2}},
\end{equation}
where the approximation becomes more accurate for sufficiently large $R_s$. It should not be so surprising that the $k$-th layer complexity of the sphere decoder is inversely proportional to the volume of the Voronoi region of the lattice generated by the partial upper triangular matrix $\pmb{R}_{kk}$. Since $\pmb{R}_{kk}$ is related to the channel matrix $\pmb{H}^c$, it is to be expected that the computational complexity depends critically on the channel conditions, i.e., depends on whether the channel is \textit{ill} or \textit{well} conditioned. We are now ready to establish our upper bound on the decoder's complexity tail distribution.

\section{Computational Complexity: Tail Distribution in the High SNR Regime}
In this section, we consider a sphere radius $R_s^2=MT(1+\zeta\log\rho)$, where $\zeta>0$ is chosen sufficiently large enough to ensure the existence of at least one lattice point inside the search sphere. The reason for that choice will become evident as we further analyze the complexity of the decoder.

In this section,  we are interested in finding an upper bound to the tail distribution of the decoder's computational complexity at high SNR. This is summarized in the following theorem:
\begin{thm} The asymptotic computational complexity distribution of the MMSE-DFE sphere decoder in the $M\times N$ LAST coded MIMO channel with codeword length $T$, is upper bounded by
\begin{equation}\label{TDN}
\Pr(C\geq L)\limi{\leq} \rho^{-\eta(r)},
\end{equation}
under the condition that
\begin{equation}\label{L}
L\geq m+ V(\mathcal{S}^m_{\pmb{0}}(2R_s))\sum\limits_{k=1}^m{V(\mathcal{S}^k_{\pmb{0}}(R_s))\over \det(\pmb{R}_{kk}^\mathsf{T}\pmb{R}_{kk})^{1/2}},
\end{equation}
where the SNR exponent $\eta(r)=(M-r)(N-r)$ for $T\geq N+M-1$. The matrix $\pmb{R}_{kk}$ is the lower $k\times k$ part of $\pmb{R}=\pmb{Q}^\mathsf{T}\pmb{MG}$.
\end{thm}

\begin{proof}
We follow the same approach that is commonly used to upper bound the decoding error probability in the quasi-static MIMO channel \cite{ZT}. By separating the outage event from the non-outage event, we obtain:
\begin{equation}\label{UU1}
\begin{split}
\Pr(C\geq L)\leq \Pr(\pmb{\alpha}\in\mathcal{O})+\Pr(C\geq L,\pmb{\alpha}\in\overline{\mathcal{O}}).
\end{split}
\end{equation}
%	When the channel is not in outage, one can verify that the asymptotic squared effective radius of the channel-code matrix can be lower bounded by $r^2_{\rm eff}(\pmb{MG})\limi>MT(1+\zeta\log\rho)$, for some $\zeta>0$. Now

Let us concentrate on bounding the second term in the RHS of (\ref{UU1}). As mentioned in Section III.B, bounding this term is considered difficult in general. However, as will be shown in the sequel, the analysis may be simplified by separating the two events that correspond to the additive noise being located inside and outside a sphere of radius $R_s$. In this case, one can upper bound the second term in the RHS of (\ref{UU1}) as follows:
\begin{equation}\label{UB_C}
\begin{split}
\Pr&(C\geq L|\overline{\mathcal{O}})\leq\Pr(C\geq L, \|\pmb{e}\|^2\leq R_s^2|\overline{\mathcal{O}})+\Pr(\|\pmb{e}\|^2\geq R_s^2).
\end{split}
\end{equation}

It is clear from the above bound that the value of the sphere radius $R_s$ affects both terms in the RHS of (\ref{UB_C}). It very important to note that $R_s$ must be selected carefully, not only to ensure the existence of a lattice point inside the sphere, but to obtain a tail distribution that is upper bounded by the outage probability of the underlying decoding scheme. The intuition behind this will be discussed later in this section. Nevertheless, in order to appropriately select the sphere radius, we study the behaviour of some parameters that correspond to the channel-code lattice $\Lambda(\pmb{MG})$ when the channel is not in outage.

First, for nested lattice codes we have that (see~\cite{Label9})
$$|\mathcal{C}(\Lambda_c,\mathcal{R})|=2^{RT}=\rho^{rT}={V(\mathcal{R})\over V_c}.$$
When the channel is not in outage, one can verify that the effective radius of the channel-code matrix, $r_{\rm eff}(\pmb{MG})$\footnote{The radius of the sphere with volume equal to $\mathcal{V}_{\pmb{0}}(\pmb{MG})$, i.e., $r_{\rm eff}(\pmb{MG})=[V(\mathcal{V}_{\pmb{0}}(\pmb{MG}))/V(\mathcal{S}^m_{\pmb{0}}(1))]^{1/m}$.}, is asymptotically given by
\begin{equation}\label{r_eff}
r_{\rm eff}(\pmb{MG})=\left[{V_c\det(\pmb{M}^\mathsf{T}\pmb{M})^{1/2}\over V(\mathcal{S}^m_{\pmb{0}}(1))}\right]^{1/m}\limi{=} MT\left[\rho^{-rT}\det(\pmb{M}^\mathsf{T}\pmb{M})^{1/2}\right]^{1/m}\limi{=}MT\rho^{\gamma},
\end{equation}
with $\gamma=[\sum_{j=1}^{M}(1-\alpha_j)^+-r]/2M>0$, when the channel is not in outage.

It is clear from (\ref{r_eff}) that, as $\rho\rightarrow\infty$ the volume of the Voronoi region $\mathcal{V}_{\pmb{0}}(\pmb{MG})$ as well as $r_{\rm eff}(\pmb{MG})$ grow quickly with SNR as $\rho^\gamma$, where $\gamma>0$ when the channel is not in outage. Hence, the sphere radius is required to increase with SNR as well in order to ensure the existing of at least one lattice point inside the decoder's search sphere. However, choosing $R_s=r_{\rm eff}(\pmb{MG})\limi{=}\rho^{\gamma}$ results into too many points inside the sphere. Therefore, $R_s$ is required to grow with SNR at slower rate than $\rho^\gamma$. For that reason, we select the search radius to be $R_s=\sqrt{MT(1+\zeta\log\rho)}$, where $\zeta>0$ (asymptotically less than $\rho^\gamma$, for all $\zeta>0$) and show that for sufficiently large $\zeta$, such radius selection guarantees (with high probability) the existing of at least one lattice point inside the sphere. This can be seen from the fact that
\begin{equation}
\begin{split}
\Pr(\|\pmb{e}\|^2> MT(1+\zeta\log\rho))\limi{\leq} \rho^{-MT\zeta},
\end{split}
\end{equation}
where the inequality follows from applying Chernoff bound. Now, for sufficiently large $\zeta$, the above probability becomes negligible. In other words, asymptotically, one can expect that the received signal is highly likely to be located inside a sphere of square radius $R_s^2=MT(1+\zeta\log\rho)$.

Next, we consider bounding the first term in the RHS of (\ref{UB_C}) from above. By viewing the decoder as a search on a tree one can interpret $C$ as the total number of nodes that have been visited by the decoder. Therefore, assuming the all-zero codeword was transmitted, one can rewrite $C$ as $C=m+\tilde{C}$, where
$$\tilde{C}=\sum_{k=1}^m \sum_{\pmb{z}_1^k\in\mathbb{Z}^k\backslash\{\pmb{0}\}}\phi(\pmb{z}_1^k).$$

Now, let $\tilde{\phi}_k(\pmb{z})$ be the indicator function defined by
$$
\tilde{\phi}_k(\pmb{x})=\begin{cases} C_k', &\text{if $\|\pmb{e}-\pmb{M}\pmb{x}\|^2\leq R_s^2$;}\cr
                                                          0, &\text{otherwise,}\end{cases}
$$
where $C_k'$ is as defined in (\ref{aprrox_bound2}). Then, one can easily verify that
$$\tilde{C}\leq \sum_{k=1}^m C_k'\sum_{\pmb{x}\in\Lambda_c^*}\phi_k(\pmb{x}),$$
where $\Lambda_c^*=\Lambda_c\backslash\{\pmb{0}\}$.
For a given lattice $\Lambda_c$, using Markov inequality, we have
\begin{equation}\label{bb}
\Pr(C\geq L|\Lambda_c)=\Pr(\tilde{C}\geq L-m|\Lambda_c)\leq{\mathbb{E}_{\pmb{e}'}\{\tilde{C}|\Lambda_c\}\over L-m},
\end{equation}
for $L>m$. Taking the expectation of $\tilde{C}$ with respect to the noise, one can easily show that\footnote{At this point, we would like to remind the reader that for the case of MMSE-DFE lattice decoding, the additive noise vector is non-Gaussian for finite $T$. However, one can show (see \cite{GC} and \cite{EZ2}) that for a well-constructed lattice the probability density function of the noise vector $\pmb{e}$, $f_{\pmb{e}}(\pmb{\nu})\leq \beta_m f_{\tilde{\pmb{e}}}(\pmb{\nu})$, where $\tilde{\pmb{e}}\sim\mathcal{N}(\pmb{0},0.5\pmb{I})$, and $\beta_m$ is a constant (has no effect at high SNR).}
\begin{equation}\label{P_C}
\begin{split}
\Pr(&C\geq L,\|\pmb{e}\|^2\leq R_s^2|\Lambda_c,\overline{\mathcal{O}})\leq{\sum_{k=1}^m C_k'\over L-m} \sum\limits_{\pmb{x}\in\Lambda_c^*}\Pr(\|\pmb{e}-\pmb{M}\pmb{x}\|^2\leq R_s^2,\|\pmb{e}\|^2\leq R_s^2|\overline{\mathcal{O}})\\
&\stackrel{(a)}{\leq}{\sum_{k=1}^m C_k'\over L-m}\sum\limits_{\pmb{x}\in\Lambda_c^*}\Pr(\|\pmb{M}\pmb{x}\|^2\leq 4R_s^2|\overline{\mathcal{O}})={\sum_{k=1}^m C_k'\over L-m}\mathsf{E}_{\pmb{M}}\left\{\sum\limits_{\pmb{x}\in\Lambda_c^*}\mathbf{1}\{\|\pmb{M}\pmb{x}\|^2\leq 4R_s^2\}\biggl|\overline{\mathcal{O}}\right\},
\end{split}
\end{equation}
where $(a)$ follows from the fact that in general one can show that for any random vectors $\pmb{u}$ and $\pmb{v}$, and $R_s>0$, it holds $\{\|\pmb{u}-\pmb{v}\|^2\leq R_s^2,\|\pmb{v}\|^2\leq R_s^2\}\subseteq \{\|\pmb{v}\|^2\leq 4R_s^2\}$, and $\mathbf{1}\{\mathcal{A}\}$ denotes the indicator function of the event $\mathcal{A}$. By taking the expectation of (\ref{P_C}) over the ensemble of random lattices (see \cite{Label9}, Theorem 4)
\begin{equation}\label{Mink11}
\Pr(C\geq L,\|\pmb{e}\|^2\leq R_s^2|\overline{\mathcal{O}})\leq{\sum_{k=1}^m C_k'\over L-m}\mathsf{E}_{\pmb{M}}\left\{{V(\mathcal{S}^m_{\pmb{0}}(2R_s))\over V_c\det(\pmb{M}^\mathsf{T}\pmb{M})^{1/2}}\biggl|\overline{\mathcal{O}}\right\}=\mathsf{E}_{\pmb{M}}\left\{\rho^{-T[\sum_{j=1}^{M}(1-\alpha_j)^+-{r}]}|\overline{\mathcal{O}}\right\},
\end{equation}
for $L\geq m+V(\mathcal{S}^m_{\pmb{0}}(2R_s))\sum_{k=1}^m C_k'$. The last inequality follows from the fact that at high SNR we have $V_c\limi{=}\rho^{-rT}$, and $\det(\pmb{M}^\mathsf{T}\pmb{M})^{1/2}\limi{=}\rho^{T\sum_{j=1}^M(1-\alpha_j)^+}$. 

Following the footsteps of \cite{GC}, averaging (\ref{Mink11}) over the channels in $\overline{\mathcal{O}}$ set, we have
\begin{equation}\label{PEs1}
\Pr(C\geq L,|\pmb{e}|^2\leq R_s^2)\limi{\leq}\int_{\overline{\mathcal{O}}}f_{\pmb{\alpha}}(\pmb{\alpha})\Pr(C\geq L,|\pmb{e}|^2\leq R_s^2|\pmb{\alpha})\;d\pmb{\alpha}\limi{\leq}\rho^{-d^*_{\rm out}(r)},
\end{equation}
where $f_{\pmb{\alpha}}(\pmb{\alpha})$ is the joint probability density function of $\pmb{\alpha}$ which, for all $\pmb{\alpha}\in\overline{\mathcal{O}}$, is asymptotically given by (see \cite{GC})
$$f_{\pmb{\alpha}}(\pmb{\alpha})\limi{=}\exp(-\log(\rho)\sum_{i=1}^M(2i-1+|N-M|)\alpha_i),$$
and $d^*_{\rm out}(r)$ is the outage SNR exponent that is given in (\ref{MMSE-DFE}).

Similar to \cite{GC}, one can show that the behaviour of the first term in the RHS of (\ref{UU1}) at high SNR is also $\rho^{-d^*_{\rm out}(r)}$. Therefore, we finally have
\begin{equation}\label{upper}
\Pr(C\geq L)\limi{\leq}\rho^{-d^*_{\rm out}(r)},
\end{equation}
under the condition that
\begin{equation*}
L\geq m+ V(\mathcal{S}^m_{\pmb{0}}(2R_s))\sum\limits_{k=1}^m{V(\mathcal{S}^k_{\pmb{0}}(R_s))\over \det(\pmb{R}_{kk}^\mathsf{T}\pmb{R}_{kk})^{1/2}}.
\end{equation*}
\end{proof}

The above results reveal that, if the number of computations performed by the decoder exceeds
\begin{equation}\label{L0}
L_0= m+ V(\mathcal{S}^m_{\pmb{0}}(2R_s))\sum\limits_{k=1}^m{V(\mathcal{S}^k_{\pmb{0}}(R_s))\over \det(\pmb{R}_{kk}^\mathsf{T}\pmb{R}_{kk})^{1/2}},
\end{equation}
then the complexity distribution of the sphere decoder is upper bounded by its outage probability. As a result, one may save on decoding complexity while still achieving near-outage performance by setting a \textit{time-out} limit at the decoder so that when the computational complexity exceeds $L_0$ the decoder terminates the search. Such time-out limit does not affect the optimal tradeoff achieved by the modified decoding scheme. To see this, suppose that the sphere decoder imposes a time-out limit so that the search is terminated once the number of computations reaches $L_0$, and hence the decoder declares an error. Let $E_s$ be the event that the decoder makes an erroneous detection when $L\leq L_0$ (this event occurs when the received signal $\pmb{y}\notin\mathcal{V}_{\pmb{x}}(\pmb{MG})$, assuming $\pmb{x}$ was transmitted). In this case, the average error probability is given by
\begin{equation}\label{Pe_C}
P_e(\rho)=\Pr(E_s \cup \{C\geq L_0\})\leq \Pr(E_s) + \Pr(C\geq L_0)\limi{\leq} \rho^{-d_{\rm out}(r)}.%+\Pr(C\geq L_0).
\end{equation}

However, since $L_0$ is random, it would be interesting to calculate the (minimum) average number of computations required by the decoder to terminate the search.

\section{Average Sphere Decoding Complexity}
It is to be expected that when the channel is ill-conditioned (i.e., in outage) the computational complexity becomes extremely large. Moreover, when the channel is in outage it is highly likely that the decoder performs an erroneous  detection. However, when the channel is \textit{not} in outage, there is still a non-zero probability that the number of computations will become large (see (\ref{Mink11})). As such, it is sometimes desirable to terminate the search even when the channel is not in outage, especially when the sphere decoder is used under practically relevant runtime constraints. Therefore, we would like to determine the \textit{minimum} average number of computations that is required by the decoder to terminate the search and declare an error without affecting the achievable DMT.

This can be expressed as
\begin{eqnarray}\label{d2}
L_{\rm out}=\mathsf{E}\{L_0(\pmb{H}^c\in\overline{\mathcal{O}})\},
\end{eqnarray}
where $L_0(\pmb{H}^c\in\overline{\mathcal{O}})$ denotes the minimum number of computations performed by the decoder when the channel is not in outage which is given in (\ref{L0}). Before we do that, we would like first to study the asymptotic (at high SNR) behaviour of $L_0$. As mentioned in Section I, due to its low encoding complexity and the ability to achieve the optimal DMT of the channel, we focus our analysis on nested LAST codes, specifically LAST codes that are generated using construction A which is described below (see \cite{Label9}). 

We consider the Loeliger ensemble of mod-$p$ lattices, where $p$ is a prime. First, we generate the set of all lattices given by
$$\Lambda_p=\kappa (\mathsf{C}+p\mathbb{Z}^{2MT})$$
where $p\rightarrow \infty$, $\kappa\rightarrow 0$ is a scaling coefficient chosen such that the fundamental volume $V_f=\kappa^{2MT}p^{2MT-1}=1$, $\mathbb{Z}_p$ denotes the field of mod-$p$ integers, and $\mathsf{C}\subset\mathbb{Z}_p^{2MT}$ is a linear code over $\mathbb{Z}_p$ with generator matrix in systematic form $[\pmb{I}\;\pmb{P}^\mathsf{T}]^\mathsf{T}$. We use a pair of self-similar lattices for nesting. We take the shaping lattice to be $\Lambda_s=\phi\Lambda_p$, where $\phi$ is chosen such that the covering radius is $1/2$ in order to satisfy the input power constraint. Finally, the coding lattice is obtained as $\Lambda_c=\rho^{-r/2M}\Lambda_s$ to satisfy the transmission rate constraint $R(\rho)=r\log\rho$. Interestingly, one can construct a generator matrix of $\Lambda_p$ as (see \cite{Conway})
\begin{equation}
\pmb{G}_p=\kappa\begin{pmatrix}
\pmb{I} & \pmb{0}\\
\pmb{P} & p\pmb{I}
 \end{pmatrix},
\end{equation}
which has a lower triangular form. In this case, one can express the generator matrix of $\Lambda_c$ as $\pmb{G}=\rho^{-r/2M}\pmb{G}'$, where $\pmb{G}'=\phi\pmb{G}_p$. Thanks to the lower triangular format of $\pmb{G}$. If $\pmb{M}$ is an $m\times m$ arbitrary full-rank matrix, and $\pmb{G}$ is an $m\times m$ lower triangular matrix, then one can easily show that
\begin{equation}\label{matrix_kk}
\det[(\pmb{MG})_{kk}]= \det(\pmb{M}_{kk})\det(\pmb{G}_{kk}),
\end{equation}
where $(\pmb{MG})_{kk}$, $\pmb{M}_{kk}$, and $\pmb{G}_{kk}$, are the lower $k\times k$ part of $\pmb{MG}$, $\pmb{M}$, and $\pmb{G}$, respectively.

Using the above result, one can express the determinant that appears in (\ref{L0}) as
\begin{equation}
\det(\pmb{R}_{kk}^\mathsf{T}\pmb{R}_{kk})=\det(\pmb{M}_{kk}^\mathsf{T}\pmb{M}_{kk})\det(\pmb{G}_{kk}^\mathsf{T}\pmb{G}_{kk})=\rho^{-rk/2M}\det(\pmb{M}_{kk}^\mathsf{T}\pmb{M}_{kk})\det({\pmb{G}'}_{kk}^\mathsf{T}{\pmb{G}'}_{kk}).
\end{equation}
Let $\mu_1\leq \mu_2\leq\cdots\leq\mu_k$ be the ordered non-zero eigenvalues of $\pmb{M}_{kk}^\mathsf{T}\pmb{M}_{kk}$, for $k=1,\cdots,m$. Then,
$$\det(\pmb{M}_{kk}^\mathsf{T}\pmb{M}_{kk})=\prod\limits_{j=1}^k\mu_j.$$
 Note that for the special case when $k=m$ we have  $\mu_{2(j-1)T+1}=\cdots=\mu_{2jT}=1+\rho\lambda_j((\pmb{H}^c)^\mathsf{H}\pmb{H}^c)$, for all $j=1,\cdots,M$.

Denote $\alpha'_i=-\log\mu_i/\log\rho$. Using (\ref{matrix_kk}), one can asymptotically express $L_0$ as
\begin{equation}
L_0=m+(\log\rho)^{m/2}\sum\limits_{k=1}^m(\log\rho)^{k/2}\rho^{c_k},
\end{equation}
where
\begin{equation}
c_k={1\over 2}\sum\limits_{j=1}^k \left({r\over M} - \alpha'_j\right)^+.
\end{equation}
Now, since $c_k$ is non-decreasing in $k$, at the high SNR regime we have
\begin{equation}
L_0=m+(\log\rho)^{m}\rho^{c_m},
\end{equation}
where
$$c_m=
T\sum\limits_{i=1}^M \left(\displaystyle{r\over M}-(1-\alpha_i)^+\right)^+.$$

The asymptotic average of $L_0$ (averaged over channel statistics) when the channel is not in outage is given by
\begin{align*}
\mathsf{E}\{L_0(\pmb{H}^c\in\overline{\mathcal{O}})\}&=\int\limits_{\pmb{\alpha}\in\overline{\mathcal{O}}}L_0 f_{\pmb{\alpha}}(\pmb{\alpha})\;d\pmb{\alpha}\\
&=m+(\log\rho)^{m}\int\limits_{\pmb{\alpha}\in\overline{\mathcal{O}}}\exp\left(\log\rho\left[T\sum\limits_{i=1}^M\left({r\over M}-(1-\alpha_i)^+\right)^+-\sum\limits_{i=1}^M (2i-1+N-M)\alpha_i\right]\right)\;d\pmb{\alpha}\\
&=m+(\log\rho)^{m}\rho^{l(r)},
\end{align*}
where $\overline{\mathcal{O}}=\left\{\pmb{\alpha}\in\mathbb{R}_+^M: \sum_{i=1}^{M}(1-\alpha_i)^{+}\geq r\right\}$, and
\begin{equation}\label{l'_r}
l(r)=\max_{\pmb{\alpha}\in\overline{\mathcal{O}}} \left[T\sum\limits_{i=1}^M\left({r\over M}-(1-\alpha_i)^+\right)^+-\sum\limits_{i=1}^M (2i-1+N-M)\alpha_i\right].
\end{equation}
It is not so difficult to see that the optimal channel coefficients that maximize (\ref{l'_r}) are
$$\alpha_i^*=1, \quad \hbox{for }i=1,\cdots,M-k,$$
and
$$\alpha_i^*=0, \quad \hbox{for }i=M-k+1,\cdots,M,$$
i.e., the same $\pmb{\alpha}^*$ that achieves the optimal DMT of the channel.
%For the case that $r$ is not an integer, say $r\in(k,k+1)$, we have
%$$\alpha_i^*=1, \quad \hbox{for }i=1,\cdots,M-k-1,$$
%$$\alpha_i^*=0, \quad \hbox{for }i=M-k+1,\cdots,M.$$
%and
%$$\alpha^*_i=k+1-r.$$
Substituting $\pmb{\alpha}^*$ in (\ref{l'_r}), we get
\begin{equation}\label{exponent_MMSE}
l(r)={Tr(M-r)\over M}-(M-r)(N-r),
\end{equation}
for $r=0,1,\cdots,M$. In this case, the asymptotic minimum average computational complexity that is required by the decoder to achieve near-optimal performance, when the channel is not in outage, can be expressed as
\begin{equation}\label{UB_MMSE}
L_{\rm out}=2MT+(\log\rho)^{2MT}\rho^{l(r)}.
\end{equation}

Some interesting remarks can be drawn from the above analysis. Consider the case of a MIMO system with arbitrary $M$, and $N$. Assuming the use of an optimal random nested LAST code of codeword length $T\geq N+M-1$ and fixed rate $R$, i.e., $r=0$. In this case, one can see that $l(0)<0$ irrespective to the value of $T$ (i.e., the average complexity is asymptotically bounded for all $T$). It is clear that the term $(\log\rho)^{2MT}\rho^{-NM}$ decays quickly to $0$ as $\rho\rightarrow\infty$. The simulation results (introduced next) agree with the above analysis (see Fig.~1).

It is clear from the above analysis that, for a given multiplexing gain $0\leq r\leq M$, the sphere decoder has much lower asymptotic average computational complexity than the exhaustive ML decoder, where the latter has decoding complexity given by $2^{RT}=\rho^{rT}$. However, there exists a \textit{cut-off} multiplexing gain, say $r_0$, such that the average computational complexity of the sphere decoder remains bounded. This should not be a surprising result, since the sphere decoder is simply viewed as a search in a tree. It is well-known that the ultimate limit to tree search decoding, such as sequential decoding \cite{MGDC}, is the computational cut-off rate, a rate above which the average number of visited nodes in the tree (i.e., computations) is unbounded. Sphere decoding algorithms are no exception. Instead of the cut-off rate we define a \textit{cut-off multiplexing gain} for the outage-limited MIMO channel. This value can be easily found by setting $l(r_0)=0$ in (\ref{UB_MMSE}), which results in
$$r_0=\biggl\lfloor{MN\over M+T}\biggr\rfloor.$$

Interestingly, if we let the number of receive antennas $N\rightarrow \infty$, then (assuming $T=N+M-1$) one can achieve a cut-off multiplexing gain $r_0=M$ which is the maximum multiplexing gain achieved by the channel. This shows that one can dramatically improve the computational complexity of the decoder by increasing the number of antennas at the receiver side. Unfortunately, it is impossible for the sphere decoder to maintain very low decoding complexity while maintaining the maximal diversity (or the optimal tradeoff) that can be achieved by the channel, especially for the case of nested LAST codes discussed previously. For the case of MMSE-DFE sphere decoding, achieving the maximum diversity $MN$ requires the use of LAST codes with codeword lengths $T\geq N+M-1$. Increasing the number of receive antennas $N$ requires increasing $T$ as well, and hence, the second term in (\ref{UB_MMSE}) does not decay very quickly to zero. It turns out that the sphere decoder may achieve \textit{linear} computational complexity $m=2MT$ (the signal dimension) at high SNR for large enough number of antennas $N$ and fixed $T$, however at the expense of losing the maximum diversity $MN$ (or losing the optimal tradeoff).

\subsection{Sphere vs. Sequential Decoding}
A more efficient decoder that is capable of achieving good error performance with much lower decoding complexity compared to the sphere decoder is the so-called lattice sequential decoder \cite{MGDC},\cite{WD}. Such a decoder, inspired by the conventional sequential decoding algorithms such as the Fano and the Stack algorithms, provides excellent performance-complexity tradeoffs through the use of a decoding parameter called the bias. It has been shown in \cite{WD} that for a small fixed bias the average decoding complexity of the MMSE-DFE lattice sequential decoder is given by
\begin{equation}\label{UB_MMSE_seq}
L_{\rm sequential}=2MT+(\log\rho)^{MT}\rho^{l(r)},
\end{equation}
where $l(r)$ is as defined in (\ref{exponent_MMSE}). For a fixed rate $R$, i.e., for $r=0$, the ratio of the average complexity of both decoders, say $\gamma$, is given by
$$\gamma={L_{\rm sphere}\over L_{\rm sequential}}={2MT+(\log\rho)^{2MT}/ \rho^{MN}\over 2MT+(\log\rho)^{MT}/\rho^{MN}}.$$
It is clear from the above ratio that sequential decoding saves on average computational complexity at high SNR, especially for large signal dimensions. For example, consider the case of a $3\times 3$ LAST coded MIMO system with $T=5$ and fixed rate. At $\rho=10^3$ (30 dB), we have $\gamma\approx 31$, i.e., the sphere decoder's complexity is about 31 times larger than the complexity of the lattice sequential decoder. As will be shown in the sequel, simulation results agree with the above theoretical analysis. For $\rho < 30$ dB, one would expect the ratio $\gamma\gg 31$. For extremely high SNR values (e.g., $\rho\gg 30$ dB), it seems that $\gamma\rightarrow 1$ as $\rho\rightarrow\infty$. However, as will be shown next, the reduction in the computational complexity of the sequential decoder comes at the price of some performance loss compared to the sphere decoder. The performance loss increases as the codeword length $T$ increases. Hence, there is a tradeoff.

\section{Simulation Results}
\subsection{Tail Distribution \& Average Complexity}
The average computational complexity for the MMSE-DFE sphere decoder is plotted in Fig.~\ref{fig:C1_MMSE} for the case of LAST coded MIMO system with $M=N=2$, and $R=4$ bits per channel use (bpcu), for different values of codeword lengths $T=3,4,5$. The figure shows how the average number of computations decays very quickly to $m=2MT$ (the signal dimension) at high SNR, even for large values of $T$. Fig.~2 demonstrates the fact that the MMSE-DFE sphere decoder has a cut-off rate such that the average complexity of the decoder remains bounded as long as we operate below it. The figure shows that for fixed $M$, $N$, and $T$, if we increase the rate, the average complexity increases as well and may become unbounded even at high SNR.

\subsection{Performance-vs-Complexity}
An example of the performance-complexity tradeoff that results in using the lattice sequential decoder instead of the sphere decoder is depicted in Fig.~\ref{fig:m30}. Here, we consider a nested LAST coded $3\times 3$ MIMO channel with $T=5$ and $R=4$ bits per channel use. One can notice the amount of computations saved by the lattice sequential decoder for all values of SNR. For example, at $\rho=30$ dB, the average complexity of the sphere decoder is about 30 times the complexity of the lattice sequential decoder for an optimal LAST coded MIMO system with dimension $m=30$. This is achieved at the expense of some loss in performance ($\sim$0.6 dB).

 Finally, in Fig.~\ref{fig:m31}, we compare the error performance of several lattice decoders that achieve the optimal tradeoff of the channel including, the MMSE-DFE sphere decoder, the MMSE-DFE sequential decoder, and the lattice-reduction aided MMSE-DFE decoder for a $3\times 3$ LAST coded MIMO channel with $T=5$, and rate $R=12$ bpcu. It is clear from the figure that, although all of the these decoders are capable of achieving the maximal diversity $9$, lattice-reduction MMSE-DFE decoder encounters a big loss in performance ($\sim 2.7$ dB) compared to the MMSE-DFE sphere decoder. For larger values of $T$, the SNR gap between the those decoders will become wider. Therefore, in order to achieve very close to the outage performance, one need to resort to a more reliable decoders such as the sphere decoder.

\section{Summary}
In this paper, we have provided a complete analysis for the computational complexity of the MMSE-DFE sphere decoder applied to the LAST coded MIMO channel, at the high SNR regime. An upper bound of the asymptotic complexity distribution has been derived. It has been shown that if the number of computations performed by the decoder exceeds a certain limit, the complexity's tail distribution becomes upper bounded by the asymptotic outage probability achieved by the LAST coding and sphere decoding schemes. As a result, the tradeoff of the MIMO channel is naturally extended to include the decoder's complexity. The average number of computations that is required to terminate the search when the channel is not in outage has been calculated in terms of the system parameters ($\rho$,$M$,$N$,$T$). In order to achieve high order diversity, the number of antennas and the codeword length must be increased simultaneously, causing the complexity of the decoding to increase. As expected, MMSE-DFE preprocessing significantly improves the overall computational complexity of the underlying decoding scheme. However, it has been shown that there exists a cut-off multiplexing gain for which the average complexity remains bounded. Finally, the performance-complexity tradeoff achieved by several decoders that perform lattice decoding have been discussed.

%+Bibliography

%-Bibliography

%\begin{figure}[ht!]
%\center
%\includegraphics[width=3.5in]{Upper_Bound.pdf}
%\caption{A geometrical approach for upper bounding the complexity distribution. Spheres of radius $R_s$ centered at the lattice points $\pmb{x}\in\Lambda(\pmb{MG})$ are presented in dashed lines. The doted line represents the decoder's search sphere centerd at the received signal $\pmb{y}$ of radius $R_s$.}
%\end{figure}

%\begin{figure}[htbp]
%\center
%\includegraphics[width=6in]{Aler_N1.pdf}
%\caption{$(a)$ Performance and $(b)$ complexity distribution achieved by the naive sphere decoder for the case of 2$\times 2$ LAST coded MIMO channel.}
%\label{fig:Aler_N1}
%\end{figure}
%\begin{figure}[htbp]
%\center
%\includegraphics[width=6in]{Aler_11.pdf}
%\caption{$(a)$ Performance and $(b)$ complexity distribution achieved by the MMSE-DFE sphere decoder for the case of 2$\times 2$ LAST coded MIMO channel.}
%\label{fig:Aler_11}
%\end{figure}

\begin{figure}[htbp]
\center
\includegraphics[width=6in]{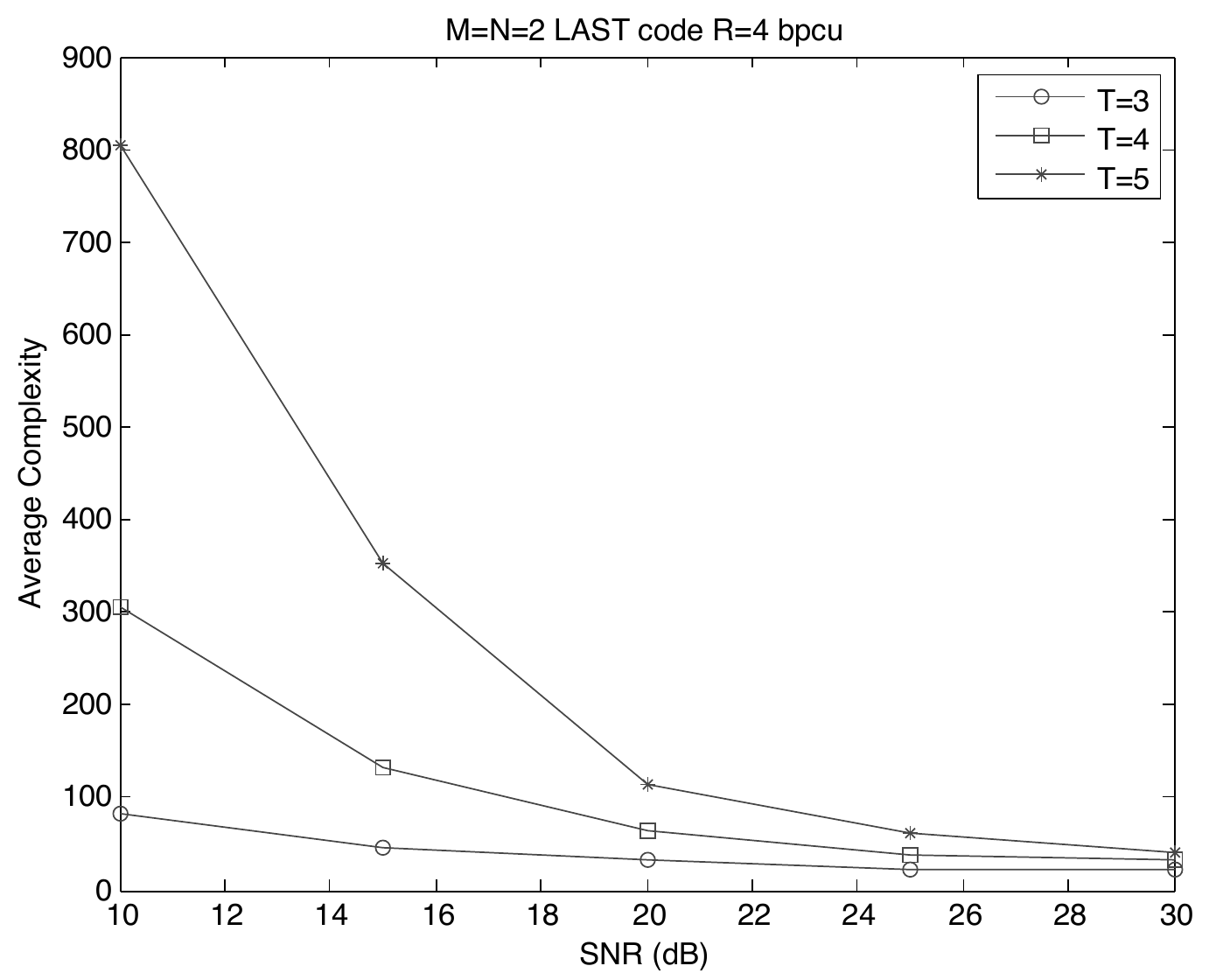}
\caption{The reduction in computational complexity achieved by the MMSE-DFE lattice decoder for all values of $T$ that achieve maximum diversity $4$. All curves decays quickly to $m=2MT=4T$ at high SNR. }
\label{fig:C1_MMSE}
\end{figure}
%
%\begin{figure}[htbp]
%\center
%\includegraphics[width=6in]{C_Naive_1}
%\caption{The computational complexity achieved by the naive lattice decoder for values of $T=1$ and $T=2$, that achieve maximum diversity $2$.}
%\label{fig:C_Naive_1}
%\end{figure}
%%
%\begin{figure}[htbp]
%\center
%\includegraphics[width=6in]{C_Naive_2}
%\caption{The computational complexity achieved by the naive lattice decoder for values of $T=3$.}
%\label{fig:C_Naive_2}
%\end{figure}

\begin{figure}[ht!]
\center
\includegraphics[width=6in]{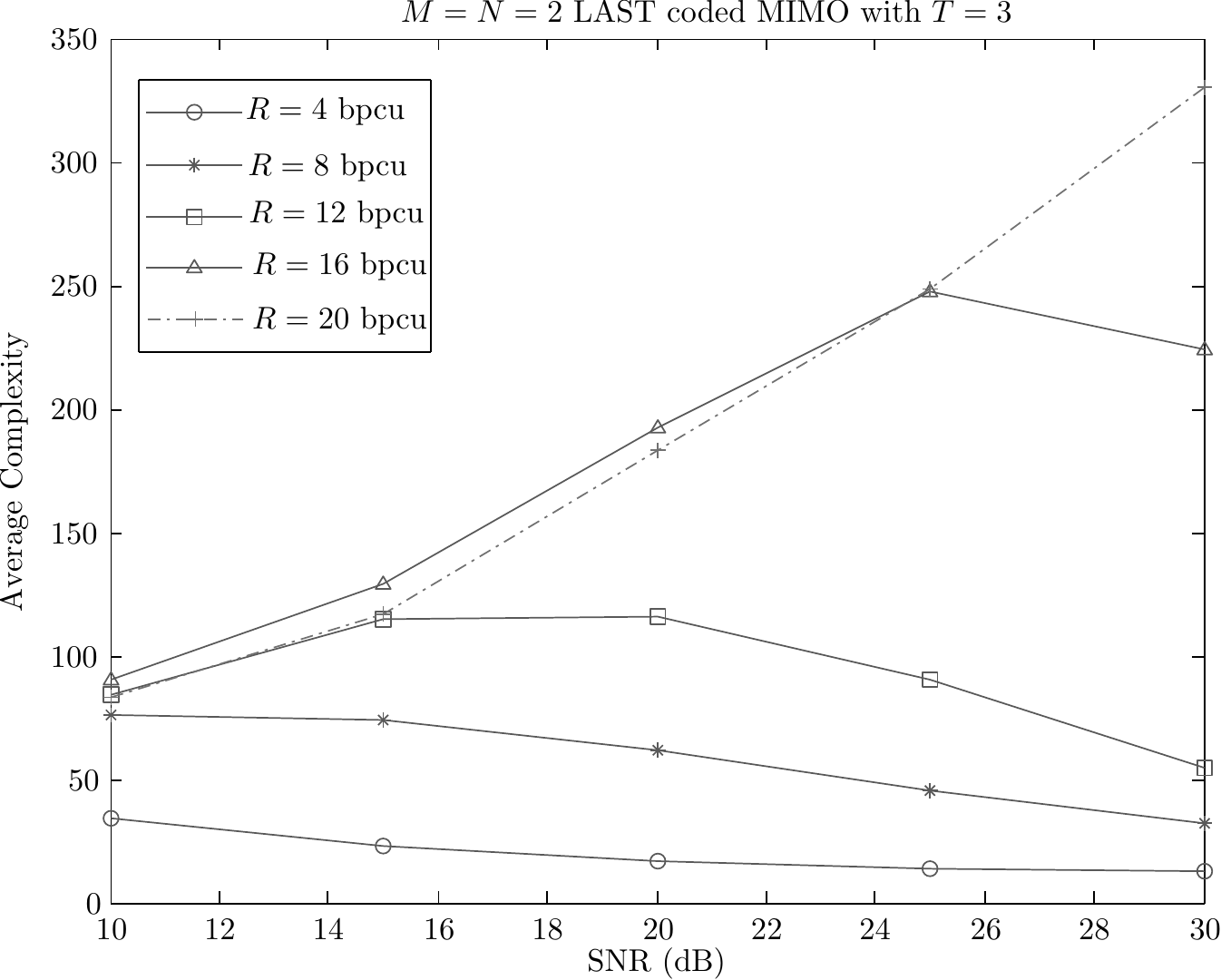}
\caption{Plots of the average complexity of the MMSE-DFE sphere decoder for an optimal nested LAST coded $2\times 2$ MIMO system with different rates $R$ in bpcu.}
\end{figure}

\begin{figure}[htbp]
\center
\includegraphics[width=7in]{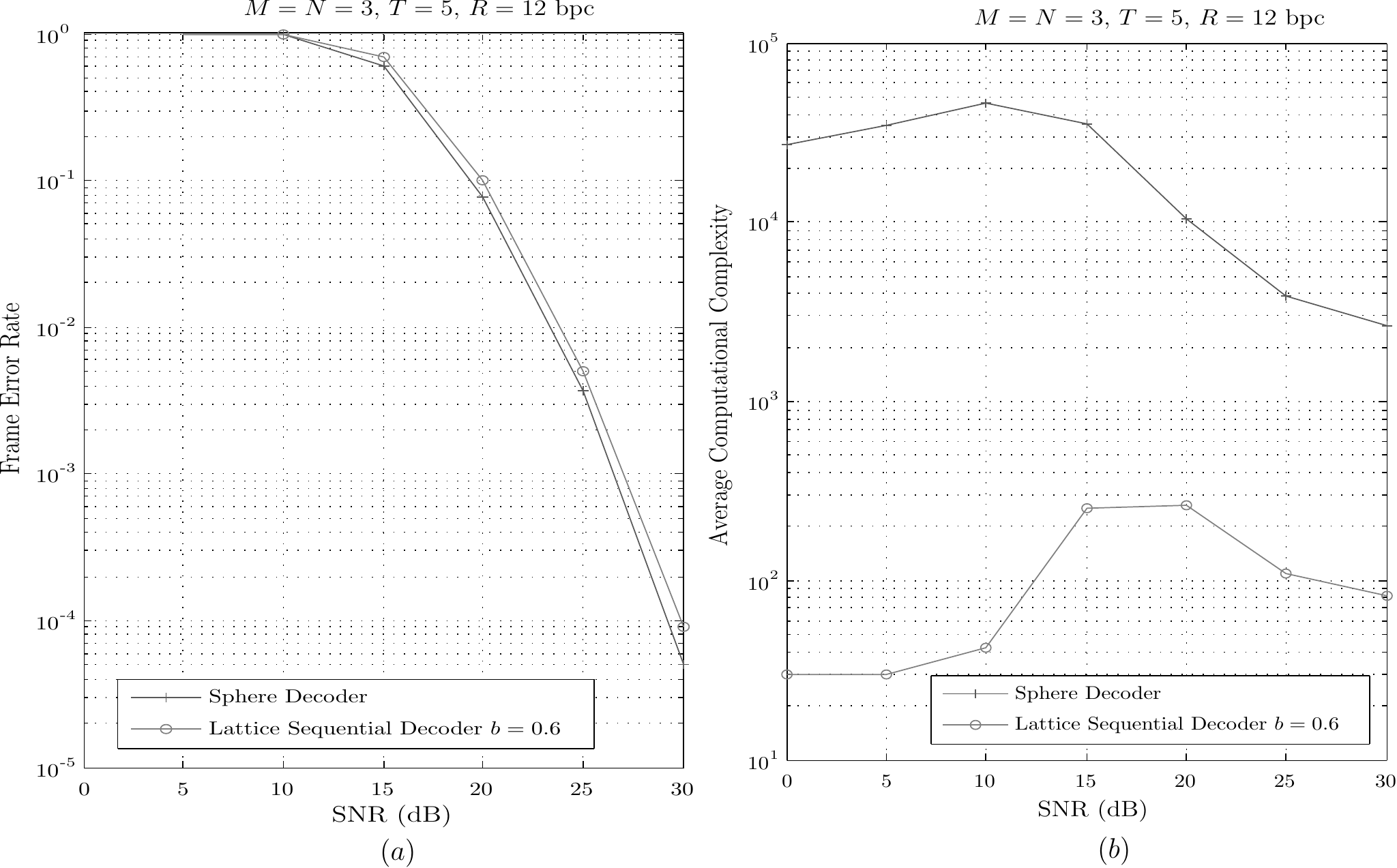}
\caption{$(a)$ Performance and $(b)$ average computational complexity comparison between sphere decoding and lattice sequential decoding for a signal with dimension $m=30$.}
\label{fig:m30}
\end{figure}

\begin{figure}[htbp]
\center
\includegraphics[width=5in]{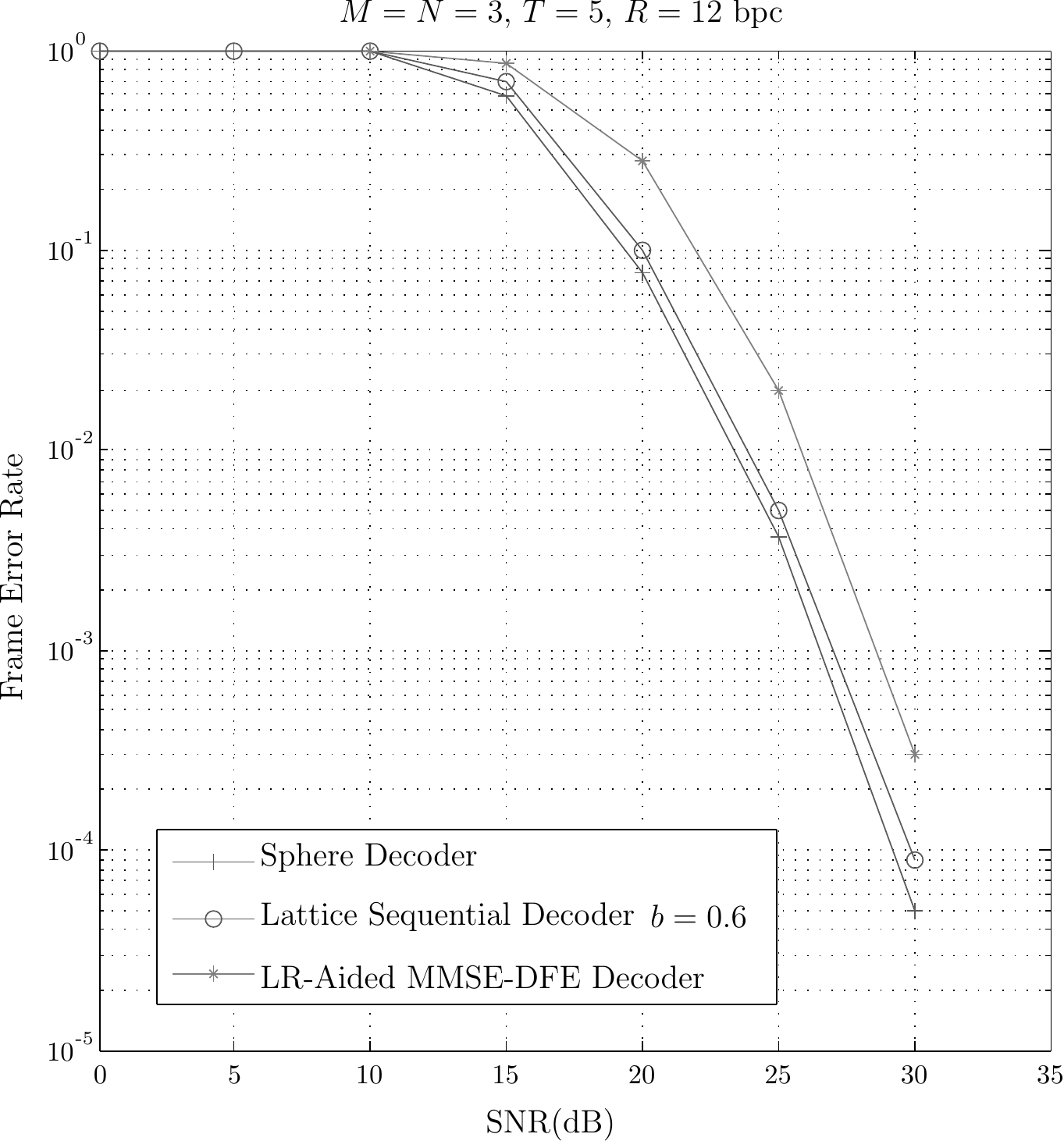}
\caption{Performance comparison between several lattice decoders that achieve the optimal DMT of the channel for a LAST coded MIMO system with dimension $m=30$.}
\label{fig:m31}
\end{figure}
\end{document}